\documentclass[12pt]{article}

\usepackage{customtemplate}

\newcommand{\Q}{\mathbb{Q}}

\def\Fq{{\mathbb F}_q}

\def\AA{{\mathbb A}}

\newcommand{\Calm}{{C^{\AA}(\ell, m)}}

\DeclareMathOperator{\ev}{Ev}
\DeclareMathOperator{\rank}{rank}
\DeclareMathOperator{\aut}{Aut}
\DeclareMathOperator{\GL}{GL}
\DeclareMathOperator{\SL}{SL}
\DeclareMathOperator{\wt}{wt}

\theoremstyle{definition}

\numberwithin{theorem}{section}

\title{The Weight Spectrum of the Affine Grassmann Code $C^{\AA}(3,6)$}
\author{Prasant Singh\footnote{Email: psinghprasant@gmail.com\\
Department of Mathematics, Indian Institute of Technology, Jammu, India-181221.} and Rohit Yadav\footnote{Email: rohityadavau1998@gmail.com\\
Department of Mathematics, Indian Institute of Technology, Jammu, India-181221.}}

\date{}

\begin{document}
\maketitle
\begin{abstract}
Affine Grassmann codes are linear codes obtained by evaluating linear combinations of minors of a generic $\ell\times\ell'$ matrix at all points of the corresponding affine space over a finite field. In this article, we determine the complete weight distribution of the affine Grassmann code $C^\AA(3,6)$. To this end, we classify the codewords according to the largest size of a minor occurring with a nonzero coefficient, and in each case we count the zeros of the corresponding polynomial on the space of $3\times3$ matrices over $\Fq$. It turns out that $C^\AA(3,6)$ has $15$ distinct nonzero weights when $q\ge3$, and $14$ when $q=2$. Unlike the case of $C^\AA(2,m)$, the weight of a codeword of $C^\AA(3,6)$ is not determined by rank data alone; for certain codewords it depends on the number of roots in $\Fq$ of an associated quadratic polynomial.
\end{abstract}



\section{Introduction}\label{sec:intro}

Let $\Fq$ denote the finite field with $q$ elements, where $q$ is a prime power. Fix positive integers $\ell$ and $\ell'$ with $\ell\le\ell'$, and set $m=\ell+\ell'$ and $\delta=\ell\ell'$. Affine Grassmann codes were introduced by Beelen, Ghorpade, and H{\o}holdt \cite{BGH2010}. Roughly speaking, the affine Grassmann code $\Calm$ is obtained by evaluating $\Fq$-linear combinations of minors of all sizes of a generic $\ell\times\ell'$ matrix at all points of the affine space $\AA^\delta(\Fq)$ of $\ell\times\ell'$ matrices over $\Fq$. When $\ell=1$, this gives the first order generalized Reed--Muller code in $\ell'$ variables, and in general these codes may be viewed as a distinct variant of generalized Reed--Muller codes. It was shown in \cite{BGH2010} that $\Calm$ is an $[n,k,d]_q$ code with
\begin{equation}\label{eq:parameters}
n=q^{\delta},\qquad k=\binom{m}{\ell},\qquad\text{and}\qquad
d=q^{\delta-\ell^2}\prod_{i=0}^{\ell-1}\bigl(q^{\ell}-q^{i}\bigr).
\end{equation}
Moreover, the minimum weight codewords of $\Calm$ were characterized and enumerated in \cite{BGH2010}, and it was shown that $\Calm$ has a large automorphism group.

Affine Grassmann codes are closely related to Grassmann codes. Recall that the Grassmann code $C(\ell,m)$ is the linear code associated with the projective system given by the Pl\"ucker embedding of the Grassmannian $G_{\ell,m}$ of all $\ell$-dimensional subspaces of $\Fq^m$. These codes were introduced by Ryan \cite{Ryan1987} over the binary field and studied by Nogin \cite{Nogin1996} over arbitrary finite fields. The subset of $G_{\ell,m}(\Fq)$ on which a fixed Pl\"ucker coordinate does not vanish is in natural bijection with $\AA^\delta(\Fq)$, and if the Pl\"ucker coordinates of its points are normalized so that this coordinate equals $1$, then the remaining Pl\"ucker coordinates are precisely the minors of all sizes of the corresponding $\ell\times\ell'$ matrix. Consequently, $\Calm$ can be obtained from $C(\ell,m)$ by puncturing on the complement of this affine open cell; see \cite{BGH2010}*{Section VII}.

Several further properties of affine Grassmann codes have been studied since their introduction. For example, Beelen, Ghorpade, and H{\o}holdt \cite{BGT2012} gave an explicit description of the dual code $\Calm^\perp$ and showed that its minimum distance is $3$ if $q>2$, and $4$ if $q=2$ and $\ell'>1$. In the same article, they showed that $\Calm$ is self-orthogonal except in three small cases, all with $\ell=1$ \cite{BGT2012}*{Theorem 14}, that both $\Calm$ and $\Calm^\perp$ are generated by their minimum weight codewords, and that the automorphism group of $\Calm$ contains a subgroup induced by the transformations $X\mapsto BXA^{-1}+U$ with $B\in\GL_\ell(\Fq)$, $A\in\GL_{\ell'}(\Fq)$ and $U\in M_{\ell\times\ell'}(\Fq)$. Ghorpade and Kaipa \cite{GK2013} determined the automorphism groups of $\Calm$, while Datta and Ghorpade \cite{DG2015} determined some initial and terminal higher weights of $\Calm$ when $\ell<\ell'$. Beelen and Pi\~nero \cite{BP2016} studied the structure of the duals of Grassmann codes and, in the same article, classified and enumerated the minimum weight codewords of the dual affine Grassmann codes. More recently, building on the work of Beelen and Singh \cite{BPP2021} on majority logic decoding of Grassmann codes, the authors, together with Pi\~nero \cite{PPR2026}, proposed a majority voting decoder for $\Calm$ over nonbinary fields that can asymptotically correct up to $\lfloor d/2^{\ell+1}\rfloor$ errors.

The determination of the weight distribution of a linear code is, in general, a difficult problem, and this is certainly the case for Grassmann codes. Nogin determined the weight spectra of the Grassmann codes $C(2,m)$ \cite{Nogin1996} and $C(3,6)$ \cite{Nogin1997}, and Kaipa and Pillai \cite{KP2013} determined the weight spectrum of $C(3,7)$. In essence, computing the weight spectrum of $C(\ell,m)$ requires a classification of the elements of $\bigwedge^{m-\ell}\Fq^m$ up to the action of $\GL_m(\Fq)$, and such a classification is available only in a few cases. For affine Grassmann codes, the only case treated so far is $\ell=2$. Pi\~nero and Singh \cite{PS2019} identified the affine open cell $\AA^{2(m-2)}(\Fq)$ with the complement of the Schubert divisor $\Omega_\alpha(2,m)$ in $G_{2,m}$, where $\alpha=(m-2,m)$, and showed that every codeword of $C(2,m)$ decomposes uniquely into a component from the corresponding Schubert code and a component from $C^\AA(2,m)$. Using this decomposition together with the classification of skew-symmetric matrices, they determined the weight spectrum of $C^\AA(2,m)$. For $\ell\ge3$, the weight spectrum of $\Calm$ does not seem to be known in any case, not even for $C^\AA(3,6)$, although the weight spectrum of the corresponding Grassmann code $C(3,6)$ has been known for several decades.

In this article, we determine the weight spectrum of $C^\AA(3,6)$. The method used by Nogin for $C(3,6)$ does not extend directly to the affine setting, since the group acting on $C^\AA(3,6)$ is smaller and the corresponding orbit structure is considerably finer. Instead, we work directly with the description of the codewords of $C^\AA(3,6)$ as linear combinations of minors of a generic $3\times3$ matrix. We classify the codewords according to the largest size of a minor occurring with a nonzero coefficient, and in each case we compute the number of zeros of the corresponding polynomial on $M_3(\Fq)$. This yields all the weights of $C^\AA(3,6)$ together with their multiplicities (Theorem~\ref{thm:full-wd}). A noteworthy feature of the result is that, in contrast to the case $\ell=2$, the weight of a codeword is not determined by rank data alone. Indeed, if the determinant occurs with a nonzero coefficient and the linear part has full rank, then the weight depends on the number of roots in $\Fq$ of the polynomial $t^2+\alpha_0t-c$, where $\alpha_0$ is the constant term and $c$ is the determinant of the coefficient matrix of the linear part. It turns out that $C^\AA(3,6)$ has $15$ distinct nonzero weights when $q\ge3$, and $14$ when $q=2$.

This article is organized as follows. In Section~\ref{sec:prelim}, we recall the construction of affine Grassmann codes and a subgroup of their automorphism group, and we prove some elementary counting lemmas. Section~\ref{sec:weight-C36} contains the computation of the weight spectrum of $C^\AA(3,6)$. We first consider codewords in which the determinant does not occur, then those in which it occurs with a nonzero coefficient, and finally we combine these results in Theorem~\ref{thm:full-wd} and Table~\ref{tab:full-wd}.

\section{Preliminaries}\label{sec:prelim}

Throughout this article, $\ell,\ell',m$ and $\delta$ are as in Section~\ref{sec:intro}. We denote by $M_{\ell\times\ell'}(\Fq)$ the set of all $\ell\times\ell'$ matrices with entries in $\Fq$, and we identify it with the affine space $\AA^\delta(\Fq)$. When $\ell=\ell'$, we write $M_\ell(\Fq)$ for $M_{\ell\times\ell}(\Fq)$. Let $X=(X_{ij})$ be an $\ell\times\ell'$ matrix whose entries are independent indeterminates over $\Fq$, and let $\Fq[X]$ denote the polynomial ring in the $\delta$ variables $X_{ij}$. As usual, the only $0\times0$ minor of $X$ is $1$. Let $\mathcal F(\ell,m)$ denote the $\Fq$-linear subspace of $\Fq[X]$ spanned by all minors of $X$. The minors of $X$ are linearly independent, and hence $\dim_{\Fq}\mathcal F(\ell,m)=\binom{m}{\ell}$; see \cite{BGH2010}*{Lemma 3}.

Fix an enumeration $P_1,P_2,\dots,P_{q^\delta}$ of $\AA^\delta(\Fq)$ and consider the evaluation map
\begin{align} \label{eq:evaluationmap}
   \ev:\mathcal F(\ell,m)&\longrightarrow\Fq^{q^\delta} \nonumber \\
f&\longmapsto c_f=\bigl(f(P_1),f(P_2),\dots,f(P_{q^\delta})\bigr).
\end{align}
The map $\ev$ is $\Fq$-linear, and it is injective since every element of $\mathcal F(\ell,m)$ has degree at most $1$ in each variable \cite{BGH2010}. The image of $\ev$ is the affine Grassmann code $\Calm$, and its parameters are given by \eqref{eq:parameters}.

We shall use a subgroup of the automorphism group of $\Calm$ described in \cites{BGH2010, BGT2012}. Let $B\in\GL_\ell(\Fq)$, $A\in\GL_{\ell'}(\Fq)$, and $U\in\AA^\delta(\Fq)$, and consider the affine transformation
\begin{align*}
\psi_{U,A,B}:\AA^\delta(\Fq)&\longrightarrow\AA^\delta(\Fq),\\
P&\longmapsto BPA^{-1}+U.
\end{align*}
Since $\psi_{U,A,B}$ is a bijection, there is a unique permutation $\sigma=\sigma_{U,A,B}$ of $\{1,\dots,q^\delta\}$ such that $\psi_{U,A,B}(P_i)=P_{\sigma(i)}$ for $1\le i\le q^\delta$. For $c=(c_1,\dots,c_{q^\delta})\in\Calm$, we write $\sigma_{U,A,B}(c)=(c_{\sigma(1)},\dots,c_{\sigma(q^\delta)})$.

\begin{lemma}\cite{BGT2012}*{Lemma 7}\label{lemma: Autgrp}
Let $B\in\GL_\ell(\Fq)$, $A\in\GL_{\ell'}(\Fq)$, and $U\in\AA^\delta(\Fq)$. Then $\sigma_{U,A,B}\in\aut(\Calm)$.
\end{lemma}

Note that $\sigma_{U,A,B}(c_f)=c_g$, where $g(X)=f(BX A^{-1}+U)$. In particular, $f$ and $g$ have the same weight. We shall use this repeatedly to replace a polynomial by another one of a more convenient form.

We conclude this section with four elementary counting lemmas. We include their proofs for the sake of completeness. In these lemmas, the letters $m$, $n$ and $r$ denote arbitrary positive integers and are not related to the parameters of $\Calm$.

\begin{lemma}\label{lem:hyperbolic}
Let $H:\Fq^m\times \Fq^m\to \Fq$ be the map defined by $H(\mathbf u,\mathbf v)=\sum_{i=1}^m u_iv_i$. For $t\in\Fq$ let $N_H(t)=\lvert\{(\mathbf u,\mathbf v)\in\Fq^{2m}: H(\mathbf u,\mathbf v)=t\}\rvert$. Then
$$ N_H(t)= \begin{cases}
q^{2m-1}+q^{m}-q^{m-1}, & t=0,\\
q^{2m-1}-q^{m-1}, & t\neq 0 .
\end{cases} $$
\end{lemma}

\begin{proof}
For convenience, write $N_m(t)=N_H(t)$. First, let $m=1$, so that $H(u,v)=uv$. Now $uv=0$ if and only if $u=0$ or $v=0$, and these two cases overlap only at $(0,0)$. Hence, $N_1(0)=2q-1$. If $t\ne0$, then $u\ne0$ and $v=u^{-1}t$ is uniquely determined, so $N_1(t)=q-1$. Both values agree with the stated formula.

Now let $m\ge2$ and write $H=H'+u_mv_m$, where $H'$ is the corresponding form in the first $2(m-1)$ variables. Splitting the solutions of $H=t$ according to the value $s$ of $H'$, we obtain
$$ N_m(t)=\sum_{s\in\Fq}N_{m-1}(s)N_1(t-s). $$
Separating the term $s=t$ from the others and using $\sum_sN_{m-1}(s)=q^{2m-2}$, we find that
$$ N_m(t)=N_{m-1}(t)(2q-1)+(q-1)\bigl(q^{2m-2}-N_{m-1}(t)\bigr)=qN_{m-1}(t)+(q-1)q^{2m-2} $$
for every $t\in\Fq$. Now suppose that the formula holds for $N_{m-1}$. If $t=0$, then
$$ N_m(0)=q\bigl(q^{2m-3}+q^{m-1}-q^{m-2}\bigr)+(q-1)q^{2m-2}=q^{2m-1}+q^{m}-q^{m-1}, $$
and if $t\ne0$, then
$$ N_m(t)=q\bigl(q^{2m-3}-q^{m-2}\bigr)+(q-1)q^{2m-2}=q^{2m-1}-q^{m-1}. $$
This completes the induction.
\end{proof}

\begin{lemma}\label{lem:QLC}
Let $n=2m+r$, and let $f=H(\mathbf y)+\ell(\mathbf z)+c$, where $\mathbf y\in\Fq^{2m}$, $H$ is as in Lemma~\ref{lem:hyperbolic}, $\mathbf z\in\Fq^r$, $\ell$ is a linear form on $\Fq^r$, and $c\in\Fq$. Let
$$ Z(f)=\{(\mathbf y,\mathbf z)\in\Fq^{2m}\times \Fq^r :f(\mathbf y,\mathbf z)=0\} $$
be the set of zeros of $f$. Then the following hold.
\begin{enumerate}
\item[(1)] If $\ell\equiv0$, then $|Z(f)|=q^rN_H(-c)$.
\item[(2)] If $\ell\not\equiv0$, then $|Z(f)|=q^{n-1}$ for every
$c\in\Fq$.
\end{enumerate}
\end{lemma}

\begin{proof}
Splitting $Z(f)$ according to the value $t=H(\mathbf y)$, we see that $|Z(f)|=\sum_{t\in\Fq}N_H(t)\cdot\#\{\mathbf z\in\Fq^r:\ell(\mathbf z)=-t-c\}$. If $\ell\equiv0$, then the inner count is $q^r$ when $t=-c$ and $0$ otherwise, and this proves (1). If $\ell\not\equiv0$, then $\ell$ takes every value in $\Fq$ exactly $q^{r-1}$ times, and hence $|Z(f)|=q^{r-1}\sum_tN_H(t)=q^{r-1}q^{2m}=q^{n-1}$. This proves (2).
\end{proof}

\begin{lemma}\label{lem:roots}
For $i=0,1,2$, let
$$ N_i=\#\{(b,d)\in\Fq^2: t^2+bt+d\text{ has exactly $i$ distinct roots in }\Fq\}. $$
Then $N_1=q$ and $N_0=N_2=q(q-1)/2$.
\end{lemma}

\begin{proof}
If $t^2+bt+d$ has a root $r\in\Fq$, then $t^2+bt+d=(t-r)(t-s)$ for a unique $s\in\Fq$. Since $\Fq[t]$ is a unique factorization domain, the monic quadratic polynomials over $\Fq$ having at least one root in $\Fq$ are in bijection with the unordered pairs $\{r,s\}$ of elements of $\Fq$, where $r=s$ is allowed, via $\{r,s\}\mapsto(t-r)(t-s)$. The $q$ pairs with $r=s$ give the polynomials with exactly one root, and the $q(q-1)/2$ pairs with $r\ne s$ give the polynomials with exactly two roots. Hence, $N_1=q$ and $N_2=q(q-1)/2$. Finally, $N_0=q^2-N_1-N_2=q(q-1)/2$.
\end{proof}

\begin{lemma}\label{lem:prescribed-det} For $n\ge1$ and $c\in\Fq$,
$$\lvert\{A\in M_n(\Fq):\det A=c\}\rvert=
\begin{cases}
|\GL_n(\Fq)|/(q-1), & c\neq0,\\
q^{n^2}-|\GL_n(\Fq)|, & c=0.
\end{cases}$$
In particular, for $c\neq0$ this number is equal to $|\SL_n(\Fq)|$.
\end{lemma}

\begin{proof}
The case $c=0$ is clear. Now let $c\ne0$ and fix $B_0\in M_n(\Fq)$ with $\det B_0=c$. The map $S\mapsto SB_0$ is a bijection from $\SL_n(\Fq)$ onto $\{A\in M_n(\Fq):\det A=c\}$, with inverse $A\mapsto AB_0^{-1}$. Since $\det:\GL_n(\Fq)\to\Fq^\times$ is a surjective homomorphism with kernel $\SL_n(\Fq)$, we have $|\SL_n(\Fq)|=|\GL_n(\Fq)|/(q-1)$.
\end{proof}

\section{Weight Spectrum of \texorpdfstring{$C^{\AA}(3,6)$}{CAA(3,6)}}
\label{sec:weight-C36}

In this section, we compute the weight spectrum of the affine Grassmann code $C^{\AA}(3,6)$. Recall that $C^{\AA}(3,6)$ is the image of the evaluation map defined in \eqref{eq:evaluationmap}. Since this evaluation map is injective, every codeword $c\in C^{\AA}(3,6)$ is of the form $c=c_f$ for a unique $f\in\mathcal F(3,6)$, where $\mathcal F(3,6)$ is the linear span of all minors, including the $0\times0$ minor, of the generic matrix $X=(X_{ij})$. Here, $c_f=\ev(f)$ denotes the codeword obtained by evaluating $f$ at the points of $M_3(\Fq)=\AA^9(\Fq)$. The dimension and minimum distance of $C^\AA(3,6)$ are given by \eqref{eq:parameters}. For $f\in\mathcal F(3,6)$, we write
$$ \wt(f):=\wt(c_f)=\bigl|\{A\in M_3(\Fq):f(A)\neq0\}\bigr|. $$
The vector space $\mathcal F(3,6)$ is naturally graded as
$$ \mathcal F(3,6)=V_0\oplus V_1\oplus V_2\oplus V_3, $$
where $V_i$ denotes the vector space spanned by all $i\times i$ minors of $X$, for $0\leq i\leq3$. In particular, $V_0=\Fq$, $V_1$ is spanned by the nine entries $X_{ij}$, $V_2$ by the nine $2\times2$ minors, and $V_3$ by $\det X$. We identify $V_1$ and $V_2$ with $M_3(\Fq)$ as follows. A linear form $\ell(X)=\sum_{i,j}\alpha_{ij}X_{ij}$ is identified with its coefficient matrix $(\alpha_{ij})\in M_3(\Fq)$. For $V_2$, we use the adjugate of $X$, namely, the transpose of its cofactor matrix: $\operatorname{adj}(X)_{ij}=(-1)^{i+j}M_{ji}(X)$, where $M_{ji}(X)$ denotes the $2\times2$ minor of $X$ obtained by deleting row $j$ and column $i$. We will use the standard identities
$$ X\operatorname{adj}(X) =\operatorname{adj}(X)X =(\det X)I_3, $$
as well as $\operatorname{adj}(U)=\det(U)U^{-1}$ for $U\in\GL_3(\Fq)$ and $\operatorname{adj}(ZW)=\operatorname{adj}(W)\operatorname{adj}(Z)$ for $Z,W\in M_3(\Fq)$.

The nine entries of $\operatorname{adj}(X)$ are, up to sign, the nine $2\times2$ minors of $X$ and therefore form a basis of $V_2$. Hence, every element of $V_2$ can be written uniquely as
\begin{equation} \label{eq:Adjrepr}
    q_2^C(X):=\sum_{1\leq i,j\leq3}
C_{ij}\operatorname{adj}(X)_{ij}
\end{equation}
for a unique matrix $C\in M_3(\Fq)$, which we call its coefficient matrix. For $C,N\in M_3(\Fq)$, let $\langle C, N\rangle:=\operatorname{tr}(C^tN)$. Then
\begin{equation}
    \label{eq:Tracefunction}
    q_2^C(X)=\langle C, \operatorname{adj}X\rangle.
\end{equation}

For example, the quadratic parts of Propositions~\ref{prop:rank1}, \ref{prop:rank2}, and \ref{prop:rank3} have coefficient matrices $E_{11}$, $E_{22}+E_{33}$, and $I_3$, of ranks $1$, $2$, and $3$, respectively, where $E_{ij}$ denotes the $3\times3$ matrix with $1$ in the $(i,j)$th position and $0$ elsewhere.

Throughout, when referring to the \emph{rank} of an element of $V_1$ or $V_2$, we mean the rank of its coefficient matrix in the above sense. For an element of $V_2$, this should not be confused with its rank as a quadratic form. For example,
$$ q_2^{E_{11}}(X)=X_{22}X_{33}-X_{23}X_{32} $$
has a coefficient matrix of rank $1$, whereas, viewed as a quadratic form in the nine variables $X_{ij}$, it is a hyperbolic quadratic form of rank $4$. The next lemma is a reformulation of the Laplace expansion of the determinant.

\begin{lemma}\label{lem:triple-product}
Let $A\in M_3(\Fq)$ have rows $r_1,r_2,r_3$. Then
$$ \det A=r_3\cdot(r_1\times r_2), $$
where, for $r_1=(a_{11},a_{12},a_{13}),\quad r_2=(a_{21},a_{22},a_{23})$, we define
$$ r_1\times r_2 := (a_{12}a_{23}-a_{13}a_{22},\, a_{13}a_{21}-a_{11}a_{23},\, a_{11}a_{22}-a_{12}a_{21}), $$
and $\cdot$ denotes the usual dot product on $\Fq^3$.
\end{lemma}

\begin{proof}
This follows from the Laplace expansion of $\det A$ along the third row.
\end{proof}

\begin{lemma}\label{lem:cross-fibre}
For pairs $(r_1,r_2)\in \Fq^3\times\Fq^3$, the following hold.
\begin{enumerate}
\item[(1)] $\lvert\{(r_1,r_2):r_1\times r_2=0\}\rvert=q^6-(q^3-1)(q^3-q)$.
\item[(2)] For a fixed nonzero $w\in\Fq^3$, we have $\lvert\{(r_1,r_2):r_1\times r_2=w\}\rvert=q(q^2-1)$. This value is the same for every nonzero $w$.
\end{enumerate}
\end{lemma}

\begin{proof}
(1) Note that $r_1\times r_2=0$ if and only if $r_1,r_2$ are linearly dependent. There are $(q^3-1)(q^3-q)$ independent pairs, so there are $q^6-(q^3-1)(q^3-q)$ dependent pairs.

(2) Fix a nonzero $w\in \Fq^3$ and define $w^\perp=\{v\in\Fq^3:v\cdot w=0\}$. If $r_1\times r_2=w$, then $r_1,r_2\in w^\perp$, since $r_i\cdot(r_1\times r_2)=\det[r_i;r_1;r_2]=0$ for $i=1,2$. Since $v\mapsto v\cdot w$ is a nonzero linear map from $\Fq^3$ to $\Fq$, its kernel $w^\perp$ has dimension $2$.

Now, choose a basis $w_1,w_2$ of $w^\perp$. Since $w_1,w_2$ are linearly independent, $w_1\times w_2\ne0$. Also, since $w_i\cdot(w_1\times w_2)= 0$ for $i=1, 2$, the vector $w_1\times w_2$ lies in $(w^\perp)^\perp$. Since the dot product on $\Fq^3$ is nondegenerate, $(w^\perp)^\perp=\Fq w$, and hence $w_1\times w_2=\lambda w$ for some $\lambda\ne0$.

If $r_1=x_1w_1+x_2w_2$ and $r_2=y_1w_1+y_2w_2$, then $r_1\times r_2=(x_1y_2-x_2y_1)\lambda w$. Thus, we have $r_1\times r_2=w$ if and only if $x_1y_2-x_2y_1=\lambda^{-1}$. By Lemma~\ref{lem:prescribed-det} with $n=2$, the number of such $(x_1,x_2,y_1,y_2)$ is $|\GL_2(\Fq)|/(q-1)=q(q^2-1)$. This does not depend on $w$.
\end{proof}

We first consider codewords $c_f\in C^\AA(3, 6)$ with $f\in V_0\oplus V_1\oplus V_2$, that is, codewords for which $\det X$ does not occur in the expression of $f$ as a linear combination of minors of $X$. The following lemma shows that it suffices to consider one quadratic part of each rank. Note that a substitution $X\mapsto AXB$ does not act on the coefficient matrices of elements of $V_2$ in the same way as on those of $V_1$, since the adjugate of $X$ is involved. Nevertheless, together with multiplication by nonzero scalars, the induced action is still transitive on the matrices of each fixed rank.

\begin{lemma}\label{lem:V2-transitive}
Let $C,\;C'\in M_3(\Fq)$ be two matrices with $\rank(C)=\rank(C^\prime)$. Then there exist $A,B\in\GL_3(\Fq)$ and a nonzero $\beta\in\Fq$ satisfying
$$ q_2^{C'}(X)=\beta\, q_2^{C}(AXB). $$
Consequently, for every nonnegative integer $W$, the number of pairs $(\alpha_0,\ell)\in V_0\times V_1$ with $\wt(q_2^{C}+\ell+\alpha_0)=W$ is equal to the number of such pairs with $\wt(q_2^{C'}+\ell+\alpha_0)=W$.
\end{lemma}

\begin{proof}
We first determine how the substitution $X\mapsto AXB$ acts on the coefficient matrix. Note that $\operatorname{adj}(AXB) = \operatorname{adj}(B)\operatorname{adj}(X)\operatorname{adj}(A)$. For arbitrary $P,Q\in M_3(\Fq)$, we have
\begin{align*}
\langle C, PNQ\rangle
&=\operatorname{tr}(C^tPNQ)\\
&=\operatorname{tr}(QC^tPN)\\
&=\operatorname{tr}\bigl((P^tCQ^t)^tN\bigr)\\
&=\langle P^tCQ^t,N\rangle.
\end{align*}
Applying this identity with $P=\operatorname{adj}(B),\ N=\operatorname{adj}(X),\ Q=\operatorname{adj}(A)$, and equation \eqref{eq:Tracefunction}, we obtain
\begin{equation}\label{eq:V2-action}
q_2^C(AXB)=q_2^{C''}(X),
\qquad
\text{where } C''=\operatorname{adj}(B)^tC\operatorname{adj}(A)^t.
\end{equation}

We next determine the matrices that can occur as $C''$ in \eqref{eq:V2-action}. For $U\in\GL_3(\Fq)$, we have $\operatorname{adj}(U)=\det(U)U^{-1}$, and hence $\operatorname{adj}(U)$ is invertible with
$$ \det\bigl(\operatorname{adj}(U)\bigr)=\det(U)^2, $$
which is a square in $\Fq^\times$. Conversely, let $V\in\GL_3(\Fq)$ be a matrix whose determinant is a square in $\Fq^\times$. If we write $\det(V)=d^2$ for some nonzero $d\in\Fq$ and set $W=V^t$, then we have $\det(W)=d^2$. Now, take $U=dW^{-1}$. Then
$$ \det(U)=d^3\det(W)^{-1}=d $$
and consequently
$$ \operatorname{adj}(U) =\det(U)U^{-1} =d(d^{-1}W)=W. $$
Thus, $\operatorname{adj}(U)^t=V$. It follows that, as $A$ and $B$ range over $\GL_3(\Fq)$, the matrices $\operatorname{adj}(A)^t$ and $\operatorname{adj}(B)^t$ range precisely over the invertible matrices with square determinant. In particular, $C''$ and $C$ have the same rank.

It therefore remains to show that, for matrices $C,C'\in M_3(\Fq)$ of the same rank, there exist $\beta\in\Fq^\times$ and $P,Q\in\GL_3(\Fq)$, both having square determinant, such that $C'=\beta PCQ$. Indeed, since
$$ \beta q_2^{C''}(X)=q_2^{\beta C''}(X), $$
such a representation of $C'$ will give the required relation
$$ q_2^{C'}(X)=\beta q_2^C(AXB). $$
Let $r$ be the common rank of $C$ and $C'$, and let $E_r=\operatorname{diag}(1,\ldots,1,0,\ldots,0)$, where the first $r$ diagonal entries are $1$. By the rank normal form of matrices, there exist $P_0,Q_0,P_1,Q_1\in\GL_3(\Fq)$ such that
$$ C=P_0E_rQ_0,\qquad C'=P_1E_rQ_1. $$
Suppose first that $r<3$. Since the third row and the third column of $E_r$ are zero, replacing $P_i$ by $P_i\operatorname{diag}(1,1,\lambda)$ and $Q_i$ by $\operatorname{diag}(1,1,\mu)Q_i$, with $\lambda,\mu\in\Fq^\times$, does not change $P_iE_rQ_i$. Thus, we may assume
$$ \det P_i=\det Q_i=1\qquad (i=0,1). $$
Consequently,
$$ C'=(P_1P_0^{-1})C(Q_0^{-1}Q_1), $$
where both factors have determinant $1$. Hence, we may take $\beta=1$.

If $r=3$, then $C$ and $C'$ are invertible. Put
$$ a=\frac{\det C'}{\det C},\qquad \beta=a,\qquad P=a^{-1}C'C^{-1},\qquad Q=I_3. $$
Then $C'=\beta PCQ$, while $\det P=a^{-2}$ is a square and $\det Q=1$. Thus, $P$ and $Q$ have square determinants.

Finally, let $A,B\in\GL_3(\Fq)$ and $\beta\in\Fq^\times$ be such that $q_2^{C'}(X)=\beta\, q_2^{C}(AXB)$. Then
$$ q_2^{C'}(X)+\ell(X)+\alpha_0=\beta\Bigl(q_2^{C}+\beta^{-1}\ell'+\beta^{-1}\alpha_0\Bigr)(AXB), $$
where $\ell'(X)=\ell(A^{-1}XB^{-1})$.
Since $X\mapsto AXB$ is a bijection of $M_3(\Fq)$ and multiplication by $\beta$ does not change the weight, both sides have the same weight. As $(\alpha_0,\ell)$ ranges over $V_0\times V_1$, so does $(\beta^{-1}\alpha_0,\beta^{-1}\ell')$. This proves the last assertion.
\end{proof}

When $\det X$ occurs with a nonzero coefficient, an analogous reduction of the linear part is possible, but the determinant of the coefficient matrix of a linear part of rank $3$ is an additional invariant; see Section~\ref{sec:cubic-present} and Lemma~\ref{lem:rank-only}.

\subsection{Hamming Weights of Linear and Quadratic Functions}
In this subsection, we consider the case $f\in V_0\oplus V_1\oplus V_2$. Since the zero codeword is the only codeword of Hamming weight $0$, we consider only nonzero codewords. We begin with the simplest case, namely $f\in V_0\oplus V_1$. The following proposition determines the weight distribution of codewords $c_f\in C^\AA(3,6)$ arising from nonzero $f\in V_0\oplus V_1$.

\begin{proposition}\label{prop:linear}
Let $f=\sum_{1\le i,j\le 3}\alpha_{ij}X_{ij}+\alpha_0\in\mathcal F(3,6)$ be nonzero. Then the following hold.
\begin{enumerate}
\item[(1)] If $\alpha_{ij}=0$ for all $i,j$ and $\alpha_0\ne0$, then $\wt(f)=q^9$. There are $q-1$ such $f$.
\item[(2)] If some $\alpha_{ij}\ne0$, then $\wt(f)=q^9-q^8$. There are $q^{10}-q$ such $f$.
\end{enumerate}
\end{proposition}

\begin{proof}
If $f$ is a nonzero function with $\alpha_{ij}=0$ for all $i,j$, then $f=\alpha_0$ is a nonzero constant and hence $\wt(f)=q^9$. This proves $(1)$. If $\alpha_{ij}\neq 0$ for some $i, j$, then $f$ is a nonconstant affine linear function on $\Fq^9$, so its zero set is an affine hyperplane of cardinality $q^8$. Hence, $\wt(f)= q^9-q^8$. Since the coefficients $\alpha_{ij}$ are not all zero and $\alpha_0$ is arbitrary, the number of such functions is $(q^9-1)q=q^{10}-q$. This proves $(2)$.
\end{proof}

Since the case $f\in V_0\oplus V_1$ has now been dealt with, we turn to the case
$$ f\in (V_0\oplus V_1\oplus V_2)\setminus(V_0\oplus V_1), $$
that is, $f$ has degree $2$ with respect to the grading of $\mathcal F(3,6)$. By Lemma~\ref{lem:V2-transitive}, it suffices to consider one quadratic part of each rank. The following proposition deals with the case in which the quadratic part of $f$ has rank $1$.

\begin{proposition}\label{prop:rank1}
Let $f=X_{22}X_{33}-X_{23}X_{32}+\sum_{1\le i,j\le3}\alpha_{ij}X_{ij}+\alpha_0 \in\mathcal F(3,6)$. Set
$$ g:=\alpha_{11}X_{11}+\alpha_{12}X_{12}+\alpha_{13}X_{13}+\alpha_{21}X_{21}+\alpha_{31}X_{31}, \qquad k:=\alpha_0+\alpha_{23}\alpha_{32}-\alpha_{22}\alpha_{33}, $$
so that $g$ consists of the terms of the linear part of $f$ involving the five variables that do not occur in $X_{22}X_{33}-X_{23}X_{32}$. Then the weight distribution of the $q^{10}$ codewords of this form is given by the following table.
$$
\begin{array}{c|c|c}
\hline
\text{Condition} & \text{Weight} & \text{Number of codewords}\\
\hline
g\neq 0 & q^8(q-1) & q^{10}-q^5\\
g=0,\ k=0 & q^6(q^2-1)(q-1) & q^4\\
g=0,\ k\neq0 & q^6(q^3-q^2+1) & q^5-q^4\\
\hline
\end{array}
$$
\end{proposition}

\begin{proof}
We have
$$ X_{22}X_{33}+\alpha_{22}X_{22}+\alpha_{33}X_{33} =(X_{22}+\alpha_{33})(X_{33}+\alpha_{22}) -\alpha_{22}\alpha_{33}, $$
and
$$ -X_{23}X_{32}+\alpha_{23}X_{23}+\alpha_{32}X_{32} =-(X_{23}-\alpha_{32})(X_{32}-\alpha_{23}) +\alpha_{23}\alpha_{32}. $$
Therefore, under the bijective change of variables
$$ X_{22}'=X_{22}+\alpha_{33},\qquad X_{33}'=X_{33}+\alpha_{22}, $$
$$ X_{23}'=X_{23}-\alpha_{32},\qquad X_{32}'=X_{32}-\alpha_{23}, $$
the polynomial $f$ becomes
$$ H+g+k, \qquad H=X_{22}'X_{33}'-X_{23}'X_{32}'. $$
Here $H$ is a hyperbolic quadratic form in four variables. Since the change of variables is a bijection, it preserves the number of zeros and hence the Hamming weight. We may therefore apply Lemma~\ref{lem:QLC} with $m=2$ and $r=5$.

If $g\neq0$, then $|Z(f)|=q^8$, giving
$$ \wt(f)=q^9-q^8=q^8(q-1). $$
If $g=0$ and $k=0$, then $|Z(f)|=q^5(q^3+q^2-q)$, and
$$ \wt(f)=q^9-q^5(q^3+q^2-q) =q^6(q-1)(q^2-1). $$
Finally, if $g=0$ and $k\neq0$, then $|Z(f)|=q^5(q^3-q)$, and therefore
$$ \wt(f)=q^9-q^5(q^3-q) =q^6(q^3-q^2+1). $$
It remains to count the codewords in each case. If $g\neq0$, then its five coefficients are not all zero, giving $q^5-1$ choices, while the remaining five coefficients are arbitrary. Hence, the number of such codewords is
$$ (q^5-1)q^5=q^{10}-q^5. $$
If $g=0$, then its five coefficients vanish, leaving $\alpha_{22},\alpha_{23},\alpha_{32},\alpha_{33}$ arbitrary. Thus, there are $q^4$ choices for these four coefficients. For each such choice, exactly one value of $\alpha_0$ gives $k=0$. Consequently, there are $q^4$ codewords with $k=0$ and
$$ q^4(q-1)=q^5-q^4 $$
codewords with $k\neq0$.
\end{proof}

The number of $3\times3$ matrices of rank $1$ over $\Fq$ is $(q^3-1)(q^2+q+1)$. By Lemma~\ref{lem:V2-transitive}, for each such matrix $C$, the $q^{10}$ polynomials $q_2^C+\ell+\alpha_0$ have the same weight distribution as the family in Proposition~\ref{prop:rank1}.

\begin{proposition}\label{prop:rank2}
Let $f\in V_0\oplus V_1\oplus V_2$ be of the form
$$ f=(X_{11}X_{22}-X_{12}X_{21})+(X_{11}X_{33}-X_{13}X_{31})+\sum_{1\le i,j\le3}\alpha_{ij}X_{ij}+\alpha_0, $$
and let $k:=\alpha_0-\alpha_{11}\alpha_{22}+\alpha_{12}\alpha_{21}+\alpha_{13}\alpha_{31}$. Then the weight distribution of the $q^{10}$ codewords of this form is given by the following table.
$$
\begin{array}{c|c|c}
\hline
\text{Condition} & \text{Weight} & \text{Number of codewords}\\
\hline
\alpha_{23}\ne0 \text{ or } \alpha_{32}\ne0 \text{ or } \alpha_{22}\ne\alpha_{33}
& q^8(q-1) & q^{10}-q^7 \\
\alpha_{23}=\alpha_{32}=0,\ \alpha_{22}=\alpha_{33},\ k=0
& q^5(q^3-1)(q-1) & q^6 \\
\alpha_{23}=\alpha_{32}=0,\ \alpha_{22}=\alpha_{33},\ k\ne0
& q^5(q^4-q^3+1) & q^6(q-1) \\
\hline
\end{array}
$$
\end{proposition}

\begin{proof}
If we set $s:=X_{22}+X_{33}$, then the change of variables $(X_{22},X_{33})\longmapsto(s,X_{33})$ is a linear bijection, with inverse $X_{22}=s-X_{33}$. Under this change of variables, we have
$$ X_{11}X_{22}+X_{11}X_{33}=X_{11}s. $$
Hence, the quadratic part of $f$ becomes $X_{11}s-X_{12}X_{21}-X_{13}X_{31}$, which, after replacing $X_{12}$ and $X_{13}$ by $-X_{12}$ and $-X_{13}$, is a form as in Lemma~\ref{lem:hyperbolic} with $m=3$ in the six variables
$$ X_{11},s,X_{12},X_{21},X_{13},X_{31}. $$
The terms involving $X_{22}$ and $X_{33}$ in the linear part transform as
$$ \alpha_{22}X_{22}+\alpha_{33}X_{33} = \alpha_{22}s+(\alpha_{33}-\alpha_{22})X_{33}. $$
Thus, the linear part splits into a component involving the above six variables,
$$ \alpha_{11}X_{11}+\alpha_{22}s +\alpha_{12}X_{12}+\alpha_{21}X_{21} +\alpha_{13}X_{13}+\alpha_{31}X_{31}, $$
and a component involving the remaining three variables,
$$ \ell:=(\alpha_{33}-\alpha_{22})X_{33} +\alpha_{23}X_{23}+\alpha_{32}X_{32}. $$
Completing the square in each of the three hyperbolic pairs, as in Proposition~\ref{prop:rank1}, transforms $f$ into $H_6+\ell+k$, where $H_6$ is a hyperbolic form of rank $6$ and
$$ k=\alpha_0-\alpha_{11}\alpha_{22} +\alpha_{12}\alpha_{21} +\alpha_{13}\alpha_{31}. $$
Moreover,
$$ \ell=0 \quad\Longleftrightarrow\quad \alpha_{23}=\alpha_{32}=0 \quad\text{and}\quad \alpha_{22}=\alpha_{33}. $$
We now apply Lemma~\ref{lem:QLC} with $m=3$ and $r=3$. If $\ell\neq0$, then $|Z(f)|=q^8$, and hence
$$ \wt(f)=q^9-q^8=q^8(q-1). $$
If $\ell=0$ and $k=0$, then $|Z(f)|=q^3(q^5+q^3-q^2)$, and therefore
$$ \wt(f) =q^9-q^3(q^5+q^3-q^2) =q^5(q^3-1)(q-1). $$
Finally, if $\ell=0$ and $k\neq0$, then $|Z(f)|=q^3(q^5-q^2)$, so
$$ \wt(f) =q^9-q^3(q^5-q^2) =q^5(q^4-q^3+1). $$
It remains to count the codewords in each of these cases. The condition $\ell=0$ is equivalent to
$$ \alpha_{23}=\alpha_{32}=0, \qquad \alpha_{22}=\alpha_{33}. $$
Thus, the six parameters $\alpha_{11},\alpha_{12},\alpha_{13},\alpha_{21},\alpha_{31},\alpha_{22}(=\alpha_{33})$ are free, giving $q^6$ choices, and together with the $q$ choices of $\alpha_0$ this gives $q^7$ coefficient tuples with $\ell=0$. Hence, $\ell\neq0$ occurs for $q^{10}-q^7=q^7(q^3-1)$ choices of the coefficients. For each fixed choice of the six parameters satisfying $\ell=0$, the quantity
$$ k=\alpha_0-\alpha_{11}\alpha_{22} +\alpha_{12}\alpha_{21} +\alpha_{13}\alpha_{31} $$
is an affine linear function of $\alpha_0$ with coefficient $1$. Therefore, exactly one of the $q$ possible values of $\alpha_0$ gives $k=0$, while the remaining $q-1$ values give $k\neq0$. Thus, there are $q^6$ codewords with $\ell=0$ and $k=0$, and $q^6(q-1)$ codewords with $\ell=0$ and $k\neq0$. This yields the weight distribution stated in the table.
\end{proof}

\begin{proposition}\label{prop:rank3}
Let $f\in V_0\oplus V_1\oplus V_2$ be of the form
$$ \begin{aligned} f={}&(X_{11}X_{22}-X_{12}X_{21})+(X_{11}X_{33}-X_{13}X_{31})+(X_{22}X_{33}-X_{23}X_{32})\\ &+\sum_{1\le i,j\le3}\alpha_{ij}X_{ij}+\alpha_0, \end{aligned} $$
and let
$$ b:=\alpha_{22}+\alpha_{33}-\alpha_{11}, \quad d:=-(\alpha_{23}\alpha_{32}-\alpha_{22}\alpha_{33} +\alpha_{12}\alpha_{21} +\alpha_{13}\alpha_{31}+\alpha_0). $$
Let $N_f\in\{0,1,2\}$ be the number of distinct roots of $t^2+bt+d$ in $\Fq$. Then
$$ \wt(f)=q^9-q^8-(N_f-1)q^4. $$
Further, if $N_0,N_1,N_2$ are as in Lemma~\ref{lem:roots}, then the weight distribution of the $q^{10}$ codewords of this form is given by the following table.
$$
\begin{array}{c|c|c}
\hline
N_f & \text{Weight} & \text{Number of codewords} \\
\hline
0 & q^9-q^8+q^4 & q^8N_0=q^9(q-1)/2\\
1 & q^9-q^8 & q^8N_1=q^9\\
2 & q^9-q^8-q^4 & q^8N_2=q^9(q-1)/2\\
\hline
\end{array}
$$
\end{proposition}

\begin{proof}
For $\mathbf a=(a_1,a_2,a_3)\in\Fq^3$, let $f_{\mathbf a}$ denote the polynomial in the six variables of the last two rows of $X$ obtained from $f$ by substituting $X_{1j}=a_j$ for $j=1,2,3$. A direct computation gives
\begin{align*}
f_{\mathbf a}
={}&X_{22}X_{33}-X_{23}X_{32}
+(-a_2+\alpha_{21})X_{21}
+(a_1+\alpha_{22})X_{22}\\
&+(-a_3+\alpha_{31})X_{31}
+(a_1+\alpha_{33})X_{33}
+\alpha_{23}X_{23}+\alpha_{32}X_{32}
+\alpha_{\mathbf a}(f),
\end{align*}
where
$$ \alpha_{\mathbf a}(f) =a_1\alpha_{11}+a_2\alpha_{12}+a_3\alpha_{13}+\alpha_0. $$
Suppose first that $a_2\neq\alpha_{21}$ or $a_3\neq\alpha_{31}$. Then at least one of $X_{21}$ and $X_{31}$ occurs with a nonzero coefficient in $f_{\mathbf a}$, and neither of them occurs in the quadratic part. Consequently, as a function on $\Fq^6$,
$$ |Z(f_{\mathbf a})|=q^5. $$
Since $a_1$ is arbitrary and $(a_2,a_3)\ne(\alpha_{21},\alpha_{31})$, there are $q(q^2-1)=q^3-q$ choices of $\mathbf a$ satisfying this condition. It remains to consider the case
$$ a_2=\alpha_{21},\qquad a_3=\alpha_{31}. $$
Writing $a_1=t$, where $t\in\Fq$ is arbitrary, the function $f_{\mathbf a}$ reduces to a function $h_t$ in the four variables $X_{22},X_{23},X_{32},X_{33}$, and
$$ |Z(f_{\mathbf a})|=q^2|Z(h_t)|. $$
Explicitly,
$$ \begin{aligned} h_t={}&X_{22}X_{33}-X_{23}X_{32}+(t+\alpha_{22})X_{22}+(t+\alpha_{33})X_{33}+\alpha_{23}X_{23}+\alpha_{32}X_{32}\\ &+t\alpha_{11}+\alpha_{12}\alpha_{21}+\alpha_{13}\alpha_{31}+\alpha_0. \end{aligned} $$
Completing the square as in the proof of Proposition~\ref{prop:rank1}, we obtain
$$ h_t=(X_{22}+t+\alpha_{33})(X_{33}+t+\alpha_{22})-(X_{23}-\alpha_{32})(X_{32}-\alpha_{23})+k_t, $$
where
$$ \begin{aligned} k_t&=-(t+\alpha_{22})(t+\alpha_{33})+\alpha_{23}\alpha_{32}+t\alpha_{11}+\alpha_{12}\alpha_{21}+\alpha_{13}\alpha_{31}+\alpha_0\\ &=-(t^2+bt+d), \end{aligned} $$
with $b$ and $d$ as in the statement. Thus, up to a bijective change of variables, $h_t=H+k_t$, where $H$ is a form as in Lemma~\ref{lem:hyperbolic} with $m=2$. By Lemma~\ref{lem:hyperbolic},
$$
|Z(h_t)|=
\begin{cases}
q^3+q^2-q,&\text{if }t^2+bt+d=0,\\
q^3-q,&\text{otherwise}.
\end{cases}
$$
Therefore,
\begin{align*}
\sum_{t\in\Fq}|Z(h_t)|
&=N_f(q^3+q^2-q)+(q-N_f)(q^3-q)\\
&=q^4+(N_f-1)q^2.
\end{align*}
Since $|Z(f)|=(q^3-q)q^5+q^2\sum_{t\in\Fq}|Z(h_t)|$, we obtain $|Z(f)| = q^8+(N_f-1)q^4$, and consequently
$$ \wt(f)=q^9-q^8-(N_f-1)q^4. $$
Thus, the three possible values $N_f=0,1,2$ give the three weights stated in the proposition.

It remains to determine the number of codewords corresponding to each value of $N_f$. Fix $(\alpha_{21},\alpha_{22},\alpha_{23},\alpha_{31})$. For these fixed values, the map
$$ (\alpha_{11},\alpha_{12},\alpha_{13},\alpha_{32},\alpha_{33},\alpha_0) \longmapsto (b,d) $$
is an $\Fq$-linear surjection from $\Fq^6$ onto $\Fq^2$. Hence, each fiber has size $q^4$. Since the four fixed parameters range independently over $\Fq$, the map from the full set of $q^{10}$ coefficient tuples to $\Fq^2$ has fibers of size $q^8$. By Lemma~\ref{lem:roots}, there are $N_i$ pairs $(b,d)\in\Fq^2$ for which $t^2+bt+d$ has exactly $i$ distinct roots in $\Fq$. Consequently, the number of codewords with $N_f=i$ is
$$ q^8N_i,\qquad i=0,1,2. $$
Using the values of $N_0,N_1,N_2$ from Lemma~\ref{lem:roots}, we obtain the weight distribution listed in the table.
\end{proof}

\subsection{Hamming Weights of Functions With Cubic Part}
\label{sec:cubic-present}

We now consider the case $f\in \mathcal F(3,6)\setminus (V_0\oplus V_1\oplus V_2)$. We first consider functions of the form
$$ f=\beta\det X+\sum_{1\le i,j\le3}\alpha_{ij}X_{ij}+\alpha_0, \qquad \beta\ne0. $$
Since multiplication by a nonzero scalar does not change the Hamming weight, we may assume, without loss of generality, that $\beta=1$. Thus, we consider
$$ f=\det X+\ell+\alpha_0, $$
where $\ell=\sum_{1\le i,j\le3}\alpha_{ij}X_{ij}\in V_1$. The case in which a quadratic part is also present will be reduced to this one in Lemma~\ref{lem:translation-elim} by means of a suitable translation $X\mapsto X+U$.

Let
$$ \Gamma:=\{(P,Q)\in\GL_3(\Fq)\times\GL_3(\Fq):\ \det P=\det Q\}. $$
Since $\Gamma$ is the kernel of the homomorphism $(P,Q)\mapsto\det P\,(\det Q)^{-1}$ from $\GL_3(\Fq)\times\GL_3(\Fq)$ to $\Fq^\times$, it is a subgroup of $\GL_3(\Fq)\times\GL_3(\Fq)$. For $(P,Q)\in\Gamma$, consider the substitution $X\mapsto PXQ^{-1}$. By Lemma~\ref{lemma: Autgrp}, it preserves Hamming weights, and moreover
$$ \det(PXQ^{-1})=\det P\,(\det Q)^{-1}\det X=\det X. $$
Now let $L$ be the coefficient matrix of $\ell$, so that $\ell(X)=\operatorname{tr}(L^tX)$. Then
$$ \ell(PXQ^{-1})=\operatorname{tr}(L^tPXQ^{-1})=\operatorname{tr}(Q^{-1}L^tPX)=\operatorname{tr}\bigl((P^tLQ^{-t})^tX\bigr), $$
where $Q^{-t}:=(Q^{-1})^t$. Thus, the substitution transforms $f=\det X+\ell+\alpha_0$ into $\det X+\ell'+\alpha_0$, where $\ell'$ has coefficient matrix $L'=P^tLQ^{-t}$. In particular, $\rank L'=\rank L$ and
$$ \det L'=\det P\,(\det Q)^{-1}\det L=\det L. $$
Hence, the substitutions coming from $\Gamma$ preserve both the rank and the determinant of the coefficient matrix of the linear part.

Conversely, these two invariants determine $L$ up to this action. Indeed, let $r=\rank L$. If $r<3$, write $L=P_0E_rQ_0$ with $P_0,Q_0\in\GL_3(\Fq)$, where $E_r$ is the diagonal matrix whose first $r$ diagonal entries are $1$ and whose remaining entries are $0$. Since the third row and the third column of $E_r$ are zero, we may replace $P_0$ by $P_0\operatorname{diag}(1,1,\lambda)$ for a suitable $\lambda\in\Fq^\times$ and assume that $\det P_0\det Q_0=1$. Then $P:=P_0^{-t}$ and $Q:=Q_0^t$ satisfy $\det P=\det Q$ and $P^tLQ^{-t}=E_r$. If $r=3$, let $c=\det L$ and $L_c=\operatorname{diag}(1,1,c)$. Then $P:=I_3$ and $Q:=(L^{-1}L_c)^{-t}$ satisfy $\det Q=\det(L^{-1}L_c)^{-1}=1=\det P$ and $P^tLQ^{-t}=L_c$. Note that, by the invariance of the determinant, $c$ cannot be normalized to $1$ by substitutions coming from $\Gamma$. It follows that, up to the action of $\Gamma$, the coefficient matrix of $\ell$ can be taken to be one of the diagonal matrices
$$ 0,\qquad E_{11},\qquad E_{11}+E_{22},\qquad \operatorname{diag}(1,1,c)\ (c\in\Fq^\times), $$
according to whether its rank is $0$, $1$, $2$, or $3$. Consequently, it suffices to consider the cases
$$ \ell=0,\qquad \ell=X_{11},\qquad \ell=X_{11}+X_{22},\qquad \ell=X_{11}+X_{22}+cX_{33}\ (c\in\Fq^\times). $$
Now we determine the weight of the function $f=\det (X) +\ell +\alpha_0$ depending on these possibilities of $\ell$ and on the value of $\alpha_0$.

\begin{proposition}\label{prop:cubic-rank0}
In the case $\ell=0$, the Hamming weight of $f=\det X +\alpha_0$ is given by
$$
\wt(f)=\begin{cases}
    |\GL_3(\Fq)|, & \alpha_0=0,\\
    q^9-|\SL_3(\Fq)|, & \alpha_0\ne0.
\end{cases}
$$
\end{proposition}

\begin{proof}
If $\alpha_0=0$, then $\wt(f)$ is the number of invertible matrices in $M_3(\Fq)$. If $\alpha_0\neq0$, then $\wt(f)$ is the number of $A\in M_3(\Fq)$ with $\det A\neq -\alpha_0$, which is $q^9-|\SL_3(\Fq)|$ by Lemma~\ref{lem:prescribed-det}.
\end{proof}

In the rest of this subsection, for $A=(a_{ij})\in M_3(\Fq)$, we denote by $r_1,r_2,r_3\in\Fq^3$ the rows of $A$, so that $\det A=r_3\cdot(r_1\times r_2)$ by Lemma~\ref{lem:triple-product}. For $S\subseteq\{1,2\}$, we write $\ell_S(A)=\sum_{i\in S}a_{ii}$. Thus,
$$ \ell_\emptyset=0,\quad \ell_{\{1\}}=a_{11} \quad \text{ and } \ell_{\{1,2\}}=a_{11}+a_{22}. $$
Note that $\ell_S$ depends only on the rows $r_1$ and $r_2$, and not on $r_3$. The next lemma will be used to compute the weight of $f=\det X +\ell+\alpha_0$ when $\ell\neq 0$.

\begin{lemma}\label{lem:ES}
For $\alpha_0\in\Fq$, let
$$ \nu_S(\alpha_0)=\lvert\{(r_1,r_2)\in\Fq^3\times \Fq^3:r_1\times r_2=0,\ \ell_S(r_1,r_2)=-\alpha_0\}\rvert. $$
Then
$$
\nu_{\{1\}}(\alpha_0)=\begin{cases}
    2q^3-q, & \alpha_0=0,\\
    q^3, & \alpha_0\ne0,
\end{cases}
\quad \text{and} \quad
\nu_{\{1,2\}}(\alpha_0)=\begin{cases}
    q^3+q^2-q, & \alpha_0=0,\\
    q^3+q^2, & \alpha_0\ne0.
\end{cases}
$$
\end{lemma}

\begin{proof}
By Lemma~\ref{lem:cross-fibre} (1), the pairs $(r_1,r_2)$ with $r_1 \times r_2 = 0$ are exactly the pairs for which $r_1$ and $r_2$ are linearly dependent. Thus, either $r_1 = 0$ and $r_2$ is arbitrary, or $r_1 \neq 0$ and $r_2 = \lambda r_1$ for a unique $\lambda \in \mathbb{F}_q$. Throughout the proof, we write $r_1 = (a_{11}, a_{12}, a_{13})$.
 
First let $S = \{1\}$, so that $\ell_S(r_1,r_2) = a_{11}$ depends only on $r_1$. If $r_1 = 0$, then $a_{11} = 0$, so the condition $\ell_S(r_1,r_2) = -\alpha_0$ holds only when $\alpha_0 = 0$, and in that case all $q^3$ choices of $r_2$ occur. If $r_1 \neq 0$, then the condition $a_{11} = -\alpha_0$ depends only on $r_1$, and once such an $r_1$ has been chosen, it holds for all $q$ choices of $r_2 = \lambda r_1$. There are $q^2 - 1$ such $r_1$ when $\alpha_0 = 0$ and $q^2$ such $r_1$ when $\alpha_0 \neq 0$. Hence, this case contributes $q(q^2-1) = q^3 - q$ when $\alpha_0 = 0$ and $q^3$ when $\alpha_0 \neq 0$. Adding the two cases gives
$$\nu_{\{1\}}(0) = q^3 + (q^3 - q) = 2q^3 - q,
\qquad
\nu_{\{1\}}(\alpha_0) = q^3 \quad (\alpha_0 \neq 0).$$
Now let $S = \{1,2\}$, so that $\ell_S(r_1,r_2) = a_{11} + a_{22}$. If $r_1 = 0$, then $a_{11} = 0$ and we need the second coordinate
$a_{22}$ of $r_2$ to equal $-\alpha_0$, while the other two coordinates
of $r_2$ are free. This gives $q^2$ choices of $r_2$, for every
$\alpha_0 \in \mathbb{F}_q$. If $r_1 \neq 0$ and $r_2 = \lambda r_1$, then $a_{22} = \lambda a_{12}$, and hence the condition becomes $a_{11} + \lambda a_{12} = -\alpha_0$. We distinguish two subcases.
\begin{itemize}
\item If $a_{12} \neq 0$, then exactly one value of $\lambda$ works, for every choice of $a_{11}, a_{13}$ and every $\alpha_0$. There are $q(q-1)q = q^2(q-1)$ such $r_1$, and hence $q^2(q-1)$ pairs.
\item If $a_{12} = 0$, the condition becomes $a_{11} = -\alpha_0$, which is independent of $\lambda$, so all $q$ values of $\lambda$ work. If $\alpha_0 \neq 0$, then $a_{11} = -\alpha_0 \neq 0$, so $r_1 \neq 0$ holds automatically and $a_{13}$ is free; this gives $q$ choices of $r_1$ and therefore $q \cdot q = q^2$ pairs. If $\alpha_0 = 0$, then $a_{11} = a_{12} = 0$, so $r_1 \neq 0$ forces $a_{13} \neq 0$; this gives $q - 1$ choices of $r_1$ and therefore $q(q-1) = q^2 - q$ pairs.
\end{itemize}
 
Adding the three contributions, we get, for $\alpha_0 \neq 0$,
$$\nu_{\{1,2\}}(\alpha_0) = q^2 + q^2(q-1) + q^2 = q^3 + q^2,$$
and, for $\alpha_0 = 0$,
$$\nu_{\{1,2\}}(0) = q^2 + q^2(q-1) + (q^2 - q) = q^3 + q^2 - q.$$
\end{proof}

In the next proposition, we determine the Hamming weight of $f=\det X + \ell +\alpha_0$ when $\ell=X_{11}$.

\begin{proposition}\label{prop:cubic-rank1}
If $f=\det X +X_{11}+\alpha_0$, then the Hamming weight of the function $f$ is given by
$$
\wt(f)=
\begin{cases}
q^9-q^8-(q^3-q)(q^3-q^2), & \alpha_0=0,\\
q^9-q^8+q^5-q^3, & \alpha_0\ne0.
\end{cases}
$$
In particular, the weight is the same for all $q-1$ nonzero values of $\alpha_0$.
\end{proposition}

\begin{proof}
By Lemma~\ref{lem:triple-product}, we have $f(A)=0$ if and only if $r_3\cdot(r_1\times r_2)=-a_{11}-\alpha_0$. Fix $(r_1,r_2)$, and let $\mathbf c=r_1\times r_2$. If $\mathbf c\ne0$, then exactly $q^2$ values of $r_3$ satisfy the equation $r_3\cdot \mathbf c= -a_{11}-\alpha_0$. If $\mathbf c=0$, then all $q^3$ choices of $r_3$ satisfy this equation when $a_{11}+\alpha_0=0$, and none of them does otherwise. Using Lemma~\ref{lem:cross-fibre} (1), we obtain
$$ \lvert\{A\in M_3(\Fq):f(A)=0\}\rvert=q^2(q^3-1)(q^3-q)+q^3\nu_{\{1\}}(\alpha_0). $$
Substituting the value of $\nu_{\{1\}}(\alpha_0)$ from Lemma~\ref{lem:ES} and subtracting from $q^9$, we obtain the stated formula for $\wt(f)$.
\end{proof}

\begin{proposition}\label{prop:cubic-rank2}
If $f=\det X +X_{11}+ X_{22}+\alpha_0$, then the Hamming weight of the function $f$ is given by
$$
\wt(f)=
\begin{cases}
q^9-q^8+q^4-q^3, & \alpha_0=0,\\
q^9-q^8-q^3, & \alpha_0\ne0.
\end{cases}
$$
In particular, the weight is the same for all $q-1$ nonzero values of $\alpha_0$.
\end{proposition}

\begin{proof}
Arguing as in the proof of Proposition~\ref{prop:cubic-rank1}, with $a_{11}+a_{22}$ in place of $a_{11}$, we obtain
$$ \lvert\{A\in M_3(\Fq): f(A)=0\}\rvert=q^2(q^3-1)(q^3-q)+q^3\nu_{\{1,2\}}(\alpha_0). $$
Substituting the value of $\nu_{\{1,2\}}(\alpha_0)$ from Lemma~\ref{lem:ES}, we obtain the stated formula for $\wt(f)$.
\end{proof}

In the next proposition, we consider the remaining case, namely $f=\det X +\ell+\alpha_0$ with $\ell= X_{11} +X_{22} + cX_{33}$ for some $c\in\Fq^\times$. Note that $\ell=\operatorname{tr}(X)$ when $c=1$.

\begin{proposition}\label{prop:cubic-rank3}
Let $c\in\Fq^\times$, $\alpha_0\in\Fq$, and let $f=\det X +X_{11}+ X_{22}+ cX_{33}+\alpha_0$. Let $N\in\{0,1,2\}$ be the number of roots in $\Fq$ of the polynomial $t^2+\alpha_0t-c$. Then the Hamming weight of the function $f$ is given by
$$ \wt(f)=q^9-q^8+(1-N)q^4-q^3. $$
In particular, for $c=1$ and $\alpha_0=0$, that is, for $f=\det X+\operatorname{tr}(X)$, we have
$$
\wt(f)=
\begin{cases}
q^9-q^8-q^4-q^3, & q\text{ odd},\\
q^9-q^8-q^3, & q\text{ even}.
\end{cases}
$$
\end{proposition}

\begin{proof}
Let $\mathbf c=r_1\times r_2$ and let $e_3=(0,0,1)$. Then
$$ \det A+a_{11}+a_{22}+ca_{33}+\alpha_0=r_3\cdot(\mathbf c+ce_3)+(a_{11}+a_{22}+\alpha_0). $$
Thus, $f(A)=0$ if and only if $r_3\cdot\mathbf c'=-(a_{11}+a_{22})-\alpha_0$, where $\mathbf c'=\mathbf c+ce_3$. As before, if $\mathbf c'\ne0$, then exactly $q^2$ values of $r_3$ satisfy this equation. If $\mathbf c'=0$, that is, if $\mathbf c=-ce_3$, then either all $q^3$ choices of $r_3$ satisfy this equation or none of them does, according to whether $a_{11}+a_{22}=-\alpha_0$. Let $\nu(\alpha_0)=\lvert\{(r_1,r_2):r_1\times r_2=-ce_3,\ a_{11}+a_{22}=-\alpha_0\}\rvert$. Using Lemma~\ref{lem:cross-fibre} (2), we obtain

\begin{align*}
\lvert\{A\in M_3(\Fq):f(A)=0\}\rvert &=q^2\bigl(q^6 - q(q^2-1)\bigr)+q^3\nu(\alpha_0)\\
&=q^8-q^5+q^3+q^3\nu(\alpha_0).
\end{align*}

It remains to compute $\nu(\alpha_0)$. Take $w=-ce_3$ in the proof of Lemma~\ref{lem:cross-fibre}(2). Here $w^\perp=\{v\in \Fq^3:v_3=0\}$, and with $w_1=e_1$, $w_2=e_2$ we get $e_1\times e_2=e_3=-c^{-1}w$, so $\lambda=-c^{-1}$. Hence, the pairs with $r_1\times r_2=-ce_3$ correspond to $r_1=(x_1,x_2,0)$ and $r_2=(y_1,y_2,0)$ with $x_1y_2-x_2y_1=-c$. Here $a_{11}=x_1$ and $a_{22}=y_2$, so the extra condition is $y_2=-x_1-\alpha_0$. Substituting these values gives $x_2y_1=c-x_1^2-x_1\alpha_0=:K(x_1)$. For fixed $x_1$, the number of pairs $(x_2,y_1)$ with $x_2y_1=K(x_1)$ is $q-1$ if $K(x_1)\ne0$, and $2q-1$ if $K(x_1)=0$. The second case occurs exactly when $x_1$ is a root of $t^2+\alpha_0t-c$, and there are $N$ such roots. Hence
$$ \nu(\alpha_0)=N(2q-1)+(q-N)(q-1)=q(N+q-1). $$
Combining these, we obtain $\lvert\{A\in M_3(\Fq):f(A)=0\}\rvert=q^8+q^3-q^4+q^4N$,
and hence
$$ \wt(f)=q^9-q^8+(1-N)q^4-q^3. $$
Finally, take $c=1$ and $\alpha_0=0$. If $q$ is odd, then $t^2-1=0$ has the two roots $\pm1$, so $N=2$. If $q$ is even, then squaring is injective on $\Fq$, so $t=1$ is the only root and $N=1$. Substituting $N=2$ or $N=1$ into the formula above gives the two cases stated.
\end{proof}

\subsection{Automorphisms and Stabilizers}

In this subsection, we compute the stabilizers of the polynomials $\det X+\ell_S+\alpha_0$, where $S\in\{\emptyset,\{1\},\{1,2\},\{1,2,3\}\}$, under the substitutions of Lemma~\ref{lemma: Autgrp}. By the orbit--stabilizer theorem, these computations determine the sizes of the orbits of these polynomials under such substitutions.

For $A,B\in\GL_3(\Fq)$ and $U\in M_3(\Fq)$, the substitution $f(X)\mapsto f(AXB+U)$ maps $\mathcal F(3,6)$ onto itself and preserves weights, by Lemma~\ref{lemma: Autgrp}. For $f\in\mathcal F(3,6)$, we set
$$ \mathrm{Stab}(f):=\{(A,B,U)\in\GL_3(\Fq)\times\GL_3(\Fq)\times M_3(\Fq): f(AXB+U)=f(X)\}. $$
Since $X\mapsto AXB$ is a bijection, replacing $X$ by $A^{-1}XB^{-1}$ shows that $(A,B,U)\in\mathrm{Stab}(f)$ if and only if
\begin{equation}\label{Equi_Ident}
f(X+U)=f(A^{-1}XB^{-1}).
\end{equation}

We now compute $\mathrm{Stab}(f)$ for $f=\det X+\ell_S+\alpha_0$, where, for $S\in\{\emptyset,\{1\},\{1,2\},\{1,2,3\}\}$, we denote by $\ell_S:=\sum_{i\in S}X_{ii}$ the corresponding partial trace.

\begin{lemma}\label{lem:u-zero}
Let $f=\det X+\ell_S+\alpha_0$. If $(A,B,U)\in\mathrm{Stab}(f)$, then $U=0$.
\end{lemma}

\begin{proof}
We first expand $\det(X+U)$. Let $X=[ X_1, X_2, X_3]$ and $U=[U_1, U_2, U_3]$, where $X_i$ and $U_i$ denote the columns of $X$ and $U$, respectively. By multilinearity,
\begin{align*}
\det(X+U)
=&\det(X_1+U_1,X_2+U_2,X_3+U_3)\\
=&\det(X_1,X_2,X_3)
+\det(U_1,X_2,X_3)
+\det(X_1,U_2,X_3)
+\det(X_1,X_2,U_3)\\
&+\det(U_1,U_2,X_3)
+\det(U_1,X_2,U_3)
+\det(X_1,U_2,U_3)
+\det(U_1,U_2,U_3).
\end{align*}

The first and last terms are $\det X$ and $\det U$, respectively. We now identify the three terms containing exactly one column of $U$. Recall that the $j$th row of $\operatorname{adj}(X)$ consists of the cofactors associated with the $j$th column of $X$. Hence, the standard cofactor formula gives
$$ \det(U_1,X_2,X_3) +\det(X_1,U_2,X_3) +\det(X_1,X_2,U_3) =\operatorname{tr}(\operatorname{adj}(X)U). $$
Similarly, the three terms containing exactly two columns of $U$ are obtained from $\det U$ by replacing one column of $U$ by the corresponding column of $X$. Applying the same cofactor formula, now to $U$, gives
$$ \det(U_1,U_2,X_3) +\det(U_1,X_2,U_3) +\det(X_1,U_2,U_3) = \operatorname{tr}(\operatorname{adj}(U)X). $$
Combining these four groups of terms yields
$$ \det(X+U)=\det X+\operatorname{tr}(\operatorname{adj}(X)U)+\operatorname{tr}(\operatorname{adj}(U)X)+\det U. $$
Now $\operatorname{tr}(\operatorname{adj}(X)U)=\sum_{i,k}\operatorname{adj}(X)_{ik}U_{ki}$ is quadratic in $X$, because every entry of $\operatorname{adj}(X)$ is a $2\times2$ minor of $X$ up to sign, whereas $\operatorname{tr}(\operatorname{adj}(U)X)$ is linear in $X$. Thus, $f(X+U)$ equals $\det X$ plus this quadratic term, plus a linear term, plus a constant. On the other hand, $f(A^{-1}XB^{-1})=\gamma\det X+\ell_S(A^{-1}XB^{-1})+\alpha_0$, where $\gamma=\det(A^{-1})\det(B^{-1})$, and this has no quadratic term. Comparing the quadratic parts of both sides of \eqref{Equi_Ident}, we obtain $\sum_{i,k}U_{ki}\operatorname{adj}(X)_{ik}=0$. Since the nine entries of $\operatorname{adj}(X)$ form a basis of $V_2$, we obtain $U_{ki}=0$ for all $i,k$, that is, $U=0$.
\end{proof}

By Lemma~\ref{lem:u-zero}, every element of $\mathrm{Stab}(f)$ has $U=0$. Thus, \eqref{Equi_Ident} becomes
$$ \det X+\ell_S(X)+\alpha_0 =\det(A^{-1})\det(B^{-1})\det X+\ell_S(A^{-1}XB^{-1})+\alpha_0. $$
Comparing the coefficients of $\det X$ on both sides, we obtain $\det A\det B=1$, and the remaining terms give $\ell_S(A^{-1}XB^{-1})=\ell_S(X)$. In particular, $\mathrm{Stab}(\det X+\ell_S+\alpha_0)=\mathrm{Stab}(\det X+\ell_S)$ for every $\alpha_0\in\Fq$, and so it suffices to compute $\mathrm{Stab}(\det X+\ell_S)$.

Write $P=A^{-1}$ and $Q=B^{-1}$, and let $E_S$ denote the diagonal matrix whose $(i,i)$th entry is $1$ if $i\in S$ and $0$ otherwise. Since $\ell_S(X)=\operatorname{tr}(E_SX)$ and $\ell_S(PXQ)=\operatorname{tr}(QE_SPX)$, we see that $\ell_S(PXQ)=\ell_S(X)$ for all $X$ if and only if $QE_SP=E_S$. Hence
$$ \mathrm{Stab}(\det X+\ell_S)\cong\{(P,Q)\in\GL_3(\Fq)^2: QE_SP=E_S,\ \det P\det Q=1\}. $$

\begin{proposition}\label{prop:stab}
The stabilizer sizes are as follows.
\begin{align*}
    |\mathrm{Stab}(\det X)|=&\frac{|\GL_3(\Fq)|^2}{q-1}, \\
    |\mathrm{Stab}(\det X+X_{11})|=&q^6(q-1)^2(q^2-1)^2,\\
    |\mathrm{Stab}(\det X+X_{11}+X_{22})|=&q^5(q-1)^2(q^2-1),\\
    |\mathrm{Stab}(\det X+\operatorname{tr}X)|=&|\GL_3(\Fq)|.
\end{align*}
\end{proposition}

\begin{proof}
We consider the four cases separately. First, let $S=\emptyset$. Then $E_S=0$, so the condition $QE_SP=E_S$ is automatic and hence
$$ \mathrm{Stab}(\det X)=\{(P,Q):\det P\det Q=1\}. $$
By Lemma~\ref{lem:prescribed-det}, this set has cardinality $|\GL_3(\Fq)|^2/(q-1)$.

Let $S=\{1\}$. In this case, $E_S=e_1e_1^t$, and the condition $QE_SP=E_S$ becomes $(Qe_1)(e_1^tP)=e_1e_1^t$. This is an equality of rank-one matrices, and it holds if and only if $Qe_1=\lambda e_1$ and $e_1^tP=\lambda^{-1}e_1^t$ for some $\lambda\ne0$. For $P$ and $Q$ of this shape, expanding $\det P$ and $\det Q$ along the first row and the first column, respectively, gives $\det P=\lambda^{-1}D_P$ and $\det Q=\lambda D_Q$, where $D_P$ and $D_Q$ are the determinants of the remaining $2\times2$ blocks. The remaining two entries in the first column of $P$ and the remaining two entries in the first row of $Q$ are arbitrary, and contribute a factor $q^2$ each. The condition $\det P\det Q=1$ becomes $D_PD_Q=1$. By Lemma~\ref{lem:prescribed-det} with $n=2$, the number of pairs of $2\times2$ matrices with $D_PD_Q=1$ is $(q-1)|\SL_2(\Fq)|^2$. Multiplying by $q^2\cdot q^2$ and by the $q-1$ choices of $\lambda$ gives $q^6(q-1)^2(q^2-1)^2$.

Now, let $S=\{1,2\}$, so that $E_S=\begin{pmatrix}I_2&0\\0&0\end{pmatrix}$. Write $P=\begin{pmatrix}P_{11}&P_{12}\\P_{21}&P_{22}\end{pmatrix}$ and $Q=\begin{pmatrix}Q_{11}&Q_{12}\\Q_{21}&Q_{22}\end{pmatrix}$ in $2+1$ block form such that $P_{11},Q_{11}$ are $2\times2$ and $P_{22},Q_{22}$ are scalars. Multiplying in block form, we obtain
$$
E_SP=\begin{pmatrix}P_{11}&P_{12}\\0&0\end{pmatrix},\qquad
QE_SP=\begin{pmatrix}Q_{11}P_{11}&Q_{11}P_{12}\\Q_{21}P_{11}&Q_{21}P_{12}\end{pmatrix}.
$$
So $QE_SP=E_S$ holds if and only if
$$ Q_{11}P_{11}=I_2,\qquad Q_{11}P_{12}=0,\qquad Q_{21}P_{11}=0,\qquad Q_{21}P_{12}=0 . $$
The first equation says that $Q_{11}$ and $P_{11}$ are inverse to each other, so both lie in $\GL_2(\Fq)$ and $P_{11}=Q_{11}^{-1}$. Since $Q_{11}$ is invertible, the second equation gives $P_{12}=0$, and since $P_{11}$ is invertible, the third gives $Q_{21}=0$. Hence
$$
P=\begin{pmatrix}Q_{11}^{-1}&0\\P_{21}&P_{22}\end{pmatrix},\qquad
Q=\begin{pmatrix}Q_{11}&Q_{12}\\0&Q_{22}\end{pmatrix},
$$
where $Q_{11}\in\GL_2(\Fq)$, $P_{21}$ and $Q_{12}$ are arbitrary, and $P_{22},Q_{22}$ are scalars. Conversely, every such pair satisfies $QE_SP=E_S$. Here $\det P\det Q=P_{22}Q_{22}$, so the condition $\det P\det Q=1$ leaves $q-1$ choices of the pair $(P_{22},Q_{22})$. Altogether, we obtain $|\GL_2(\Fq)|\cdot q^2\cdot q^2\cdot(q-1)=q^5(q-1)^2(q^2-1)$.

Finally, let $S=\{1,2,3\}$. Then $E_S=I_3$, and the condition $QE_SP=E_S$ becomes $Q=P^{-1}$. In this case, $\det P\det Q=1$ holds automatically. Hence, the stabilizer is in bijection with $\GL_3(\Fq)$ via $P\mapsto(P,P^{-1})$.
\end{proof}

\subsection{The Complete Weight Distribution}

We now combine the results obtained above to determine the weight distribution of all the $q^{20}=|\mathcal F(3,6)|$ codewords of $C^\AA(3,6)$. For $0\le r\le 3$, let $\rho_r$ denote the number of matrices of rank $r$ in $M_3(\Fq)$. It is well known that
$$ \rho_0=1,\quad \rho_1=(q^3-1)(q^2+q+1), $$
$$ \rho_2=q(q-1)^2(q+1)(q^2+q+1)^2,\quad \rho_3=|\GL_3(\Fq)|, $$
and $\rho_0+\rho_1+\rho_2+\rho_3=q^9$.

The next lemma shows that, when the coefficient matrix of the linear form $\ell$ has rank less than $3$, the weight of $f=\det X+\ell+\alpha_0$ depends only on this rank and on whether $\alpha_0$ is zero. This allows us to use one representative for each such rank, as in the case where $\det X$ does not occur. When the coefficient matrix has rank $3$, the weight depends also on its determinant and on the value of $\alpha_0$. Recall that the rank of an element of $V_1$ means the rank of its coefficient matrix.

\begin{lemma}\label{lem:rank-only}
Let $\ell,\tilde\ell\in V_1$ have coefficient matrices $L,\tilde L$ of the same rank $r$, and let $\alpha_0\in\Fq$.
\begin{enumerate}
\item[(1)] If $r<3$, then
$$ \wt(\det X+\ell+\alpha_0)=\wt(\det X+\tilde\ell+\alpha_0). $$
Moreover, for fixed $\ell$, the weight of $\det X+\ell+\alpha_0$ is the same for all $\alpha_0\in\Fq^\times$.
\item[(2)] If $r=3$, then
$$ \wt(\det X+\ell+\alpha_0)=q^9-q^8+(1-N)q^4-q^3, $$
where $N$ is the number of roots in $\Fq$ of the polynomial $t^2+\alpha_0t-\det L$.
\end{enumerate}
\end{lemma}

\begin{proof}
For $A,B\in\GL_3(\Fq)$, let $\gamma(A,B)=\det(A^{-1})\det(B^{-1})$, so that the substitution $X\mapsto A^{-1}XB^{-1}$ carries $\det X$ to $\gamma(A,B)\det X$. Since the coefficient matrix of $\ell(AXB)$ is $A^{t}LB^{t}$, the group $\GL_3(\Fq)\times\GL_3(\Fq)$ acts on these coefficient matrices in the usual way, and the action is transitive on the matrices of rank $r$.

Let $r<3$, and let $\ell_S$ be the diagonal representative of rank $r$, i.e., $S$ is a proper subset of $\{1,2,3\}$ and $\ell_S=\sum_{i\in S}X_{ii}$. Pick some $j\notin S$. For $e\in\Fq^\times$, let $A=I$ and $B^{-1}=I+(e-1)E_{jj}$. This substitution only rescales column $j$ of $X$, but that column does not appear in the expression $\ell_S$. Therefore, $\ell_S(A^{-1}XB^{-1})=\ell_S(X)$, while $\gamma(A,B)=e$. Thus, every element of $\Fq^\times$ occurs as $\gamma(A,B)$ for some $A,B$ with $\ell_S(A^{-1}XB^{-1})=\ell_S(X)$.

Now, let $\ell$ and $\tilde\ell$ have common rank $r<3$. By transitivity, fix $A_0,B_0$ with $\ell(A_0^{-1}XB_0^{-1})=\ell_S(X)$, and put $\gamma_0=\gamma(A_0,B_0)$. Also fix $A_2,B_2$ with $\ell_S(A_2^{-1}XB_2^{-1})=\tilde\ell(X)$, and put $\gamma_2=\gamma(A_2,B_2)$. By the observation above, we can choose $A_1,B_1$ that fix $\ell_S$ and satisfy $\gamma(A_1,B_1)=(\gamma_0\gamma_2)^{-1}$. Since $\gamma$ is multiplicative with respect to the composition of substitutions, composing the three substitutions gives $A,B$ with $\ell(A^{-1}XB^{-1})=\tilde\ell(X)$ and $\gamma(A,B)=1$. Hence, $(\det X+\ell+\alpha_0)(A^{-1}XB^{-1})=\det X+\tilde\ell(X)+\alpha_0$, and therefore $\wt(\det X+\ell+\alpha_0)=\wt(\det X+\tilde\ell+\alpha_0)$.

It remains to show that, for fixed $\ell$, the weight does not depend on the nonzero value of $\alpha_0$. By what we have just proved, we may assume that $\ell=\ell_S$. Let $c\in\Fq^\times$, and take $A_1,B_1$ that fix $\ell_S$ with $\gamma(A_1,B_1)=c$. Put $g=\det X+\ell_S+\alpha_0$. Then $g(A_1^{-1}XB_1^{-1})=c\bigl(\det X+\ell_S(X)/c+\alpha_0/c\bigr)$, so
$$ \wt(\det X+\ell_S+\alpha_0)=\wt(\det X+\ell_S/c+\alpha_0/c). $$
Now $\ell_S/c$ has coefficient matrix $L_S/c$, of the same rank as $L_S$. Hence, the first part of the proof, applied with the constant term $\alpha_0/c$, gives $\wt(\det X+\ell_S/c+\alpha_0/c)=\wt(\det X+\ell_S+\alpha_0/c)$. Consequently, $\wt(\det X+\ell_S+\alpha_0)=\wt(\det X+\ell_S+\alpha_0/c)$ for every $c\in\Fq^\times$. As $c$ ranges over $\Fq^\times$, so does $\alpha_0/c$, and therefore the weight is the same for every nonzero $\alpha_0$. This proves part (1).

Now, let $r=3$, and put $c=\det L\ne0$. Let $\ell_c=X_{11}+X_{22}+cX_{33}$, whose coefficient matrix $L_c=\operatorname{diag}(1,1,c)$ also has determinant $c$. The substitution $X\mapsto A^{-1}XB^{-1}$ carries $\ell$ to a linear form with coefficient matrix $A^{-t}LB^{-t}$, whose determinant is $\gamma(A,B)\det L$. Take $A=I$ and choose $B$ such that $B^{-t}=L^{-1}L_c$. Then $A^{-t}LB^{-t}=L_c$, and $\gamma(A,B)=\det(L^{-1}L_c)=1$. Hence, $(\det X+\ell+\alpha_0)(A^{-1}XB^{-1})=\det X+\ell_c(X)+\alpha_0$, and therefore $\wt(\det X+\ell+\alpha_0)=\wt(\det X+\ell_c+\alpha_0)$. Part (2) now follows from Proposition~\ref{prop:cubic-rank3}.
\end{proof}

\begin{lemma}\label{lem:translation-elim}
For every nonnegative integer $W$,
$$
\begin{aligned}
&\left|\left\{(C,\ell,\alpha_0)\in M_3(\Fq)\times V_1\times\Fq:
\wt(\det X+q_2^C+\ell+\alpha_0)=W\right\}\right|\\
&\qquad=q^9\left|\left\{(\tilde\ell,\tilde\alpha_0)\in V_1\times\Fq:
\wt(\det X+\tilde\ell+\tilde\alpha_0)=W\right\}\right|.
\end{aligned}
$$
In other words, the number of elements of the form $\det X+q_2^C+\ell+\alpha_0$ having weight $W$ is $q^9$ times the number of elements of the form $\det X+\tilde\ell+\tilde\alpha_0$ having weight $W$.
\end{lemma}

\begin{proof} 
From the proof of Lemma~\ref{lem:u-zero},
$$ \det(X+U)=\det X+\operatorname{tr}(\operatorname{adj}(X)U) +\operatorname{tr}(\operatorname{adj}(U)X)+\det U. $$
Also $\operatorname{adj}(X+U)=\operatorname{adj}(X)+P(X,U)+\operatorname{adj}(U)$, where, for $\{j_1,j_2\}=\{1,2,3\}\setminus\{j\}$ ($j_1<j_2$) and $\{i_1,i_2\}=\{1,2,3\}\setminus\{i\}$ ($i_1<i_2$),
$$
P(X,U)_{ij}=(-1)^{i+j}\left[
\begin{vmatrix}
X_{j_1i_1}&X_{j_1i_2}\\
U_{j_2i_1}&U_{j_2i_2}
\end{vmatrix}
+
\begin{vmatrix}
U_{j_1i_1}&U_{j_1i_2}\\
X_{j_2i_1}&X_{j_2i_2}
\end{vmatrix}
\right].
$$
Each entry of $P(X,U)$ is linear in $X$ for fixed $U$. Applying $\langle C,-\rangle$ gives
$$ q_2^C(X+U)=q_2^C(X)+\langle C,P(X,U)\rangle+q_2^C(U). $$
Writing $f=\det X+q_2^C+\ell+\alpha_0$, this gives
\begin{align*}
    f(X+U)=&\det(X+U)+q_2^C(X+U)+\ell(X+U)+\alpha_0 \\
=&\det X+\bigl[\operatorname{tr}(\operatorname{adj}(X)U)+q_2^C(X)\bigr]
+\tilde\ell(X)+\tilde\alpha_0,
\end{align*}
where $\tilde\ell(X)=\operatorname{tr}(\operatorname{adj}(U)X) +\langle C,P(X,U)\rangle+\ell(X)$, and $\tilde\alpha_0=\det U+q_2^C(U)+\ell(U)+\alpha_0$.

Since $\operatorname{tr}(\operatorname{adj}(X)U)=\sum_{i,k}U_{ki}\operatorname{adj}(X)_{ik}=q_2^{U^t}(X)$, the bracketed term vanishes if and only if $U=-C^t$. For this choice of $U$, we have $f(X+U)=\det X+\tilde\ell(X)+\tilde\alpha_0$, and since $X\mapsto X+U$ is a bijection of $M_3(\Fq)$, we obtain $\wt(f)=\wt(\det X+\tilde\ell+\tilde\alpha_0)$.

Now fix $C\in M_3(\Fq)$ and put $U=-C^t$. Then $\tilde\ell=\ell+\ell_C$ and $\tilde\alpha_0=\alpha_0+\ell(U)+\kappa_C$, where the linear form $\ell_C$ and the constant $\kappa_C$ depend only on $C$. Hence, the map $(\ell,\alpha_0)\mapsto(\tilde\ell,\tilde\alpha_0)$ is a bijection of $V_1\times\Fq$, whose inverse is given by $\ell=\tilde\ell-\ell_C$ and $\alpha_0=\tilde\alpha_0-\ell(U)-\kappa_C$. Therefore, for each $C$, the number of pairs $(\ell,\alpha_0)$ with $\wt(\det X+q_2^C+\ell+\alpha_0)=W$ is equal to the number of pairs $(\tilde\ell,\tilde\alpha_0)$ with $\wt(\det X+\tilde\ell+\tilde\alpha_0)=W$. Summing over the $q^9$ matrices $C\in M_3(\Fq)$, we obtain the lemma.
\end{proof}

\begin{lemma}\label{lem:M-counts}
For $i=0,1,2$, let
$$ M_i=\#\{(c,\alpha_0)\in\Fq^\times\times\Fq:\ t^2+\alpha_0t-c\text{ has exactly $i$ distinct roots in }\Fq\}. $$
Then
$$ M_0=\frac{q(q-1)}{2},\qquad M_1=q-1,\qquad M_2=\frac{(q-1)(q-2)}{2}. $$
\end{lemma}

\begin{proof}
The map $(c,\alpha_0)\mapsto(b,d)=(\alpha_0,-c)$ is a bijection from $\Fq^\times\times\Fq$ onto the set of pairs $(b,d)\in\Fq^2$ with $d\ne0$. By Lemma~\ref{lem:roots}, it is therefore enough to subtract, from $N_0,N_1,N_2$, the contribution of the pairs $(b,0)$. For $d=0$, the polynomial $t^2+bt=t(t+b)$ has exactly one root if $b=0$ and exactly two roots if $b\ne0$. Hence
$$ M_0=N_0=\frac{q(q-1)}{2},\qquad M_1=N_1-1=q-1, $$
and
$$ M_2=N_2-(q-1)=\frac{(q-1)(q-2)}{2}. $$
This completes the proof.
\end{proof}

\begin{theorem}\label{thm:full-wd}
Every element of $\mathcal F(3,6)$ belongs to exactly one of the types listed in Table~\ref{tab:full-wd}. The Hamming weight and the number of codewords of each type are given in the table.
\end{theorem}

\begin{table}[!t]
\centering
\caption{Weight distribution of $C^{\AA}(3,6)$.}
\label{tab:full-wd}
\renewcommand{\arraystretch}{1.3}
\small
\begin{tabular}{p{0.25\textwidth}|p{0.37\textwidth}|p{0.30\textwidth}}
\hline
\textbf{Weight} & \textbf{Number of codewords} & \textbf{Type}\\
\hline
\multicolumn{3}{c}{$\beta=0$ (no cubic term)}\\
\hline
$0$ & $1$ & $C=0,\ \ell=0,\ \alpha_0=0$\\
$q^9$ & $q-1$ & $C=0,\ \ell=0,\ \alpha_0\ne0$\\
$q^9-q^8$ & $(q^{10}-q)+\rho_1(q^{10}-q^5)+\rho_2(q^{10}-q^7)+\rho_3q^9$ & $C=0,\ell\ne0$;\ $\rank C=2$;\ $\rank C=1,g\ne0$;\ $\rank C=3,N_f=1$\\
\hline
$q^6(q^2-1)(q-1)$ & $\rho_1q^4$ & $\rank C=1,\ g=0,\ k=0$\\
$q^6(q^3-q^2+1)$ & $\rho_1(q^5-q^4)$ & $\rank C=1,\ g=0,\ k\ne0$\\
\hline
$q^5(q^3-1)(q-1)$ & $\rho_2q^6$ & $\rank C=2,\ k=0$\\
$q^5(q^4-q^3+1)$ & $\rho_2q^6(q-1)$ & $\rank C=2,\ k\ne0$\\
\hline
$q^9-q^8+q^4$ & $\rho_3q^9(q-1)/2$ & $\rank C=3,\ N_f=0$\\
$q^9-q^8-q^4$ & $\rho_3q^9(q-1)/2$ & $\rank C=3,\ N_f=2$\\
\hline
\multicolumn{3}{c}{$\beta\ne0$ (cubic term present)}\\
\hline
$|\GL_3(\Fq)|$ & $(q-1)q^9$ & $\rank \ell=0,\ \alpha_0=0$\\
$q^9-|\SL_3(\Fq)|$ & $(q-1)^2q^9$ & $\rank \ell=0,\ \alpha_0\ne0$\\
\hline
$q^9-q^8-(q^3-q)(q^3-q^2)$ & $(q-1)q^9\rho_1$ & $\rank \ell=1,\ \alpha_0=0$\\
$q^9-q^8+q^5-q^3$ & $(q-1)^2q^9\rho_1$ & $\rank \ell=1,\ \alpha_0\ne0$\\
\hline
$q^9-q^8+q^4-q^3$ & $(q-1)q^9\rho_2+\frac{q^{10}(q-1)}{2}\rho_3$ & $\rank \ell=2,\ \alpha_0=0$;\ $\rank \ell=3,\ N=0$\\
\hline
$q^9-q^8-q^3$ & $(q-1)^2q^9\rho_2+(q-1)q^9\rho_3$ & $\rank \ell=2,\ \alpha_0\ne0$;\ $\rank \ell=3,\ N=1$\\
\hline
$q^9-q^8-q^4-q^3$ & $\frac{(q-1)(q-2)}{2}q^9\rho_3$ & $\rank \ell=3,\ N=2$\\
\hline
\end{tabular}

\medskip
\begin{minipage}{0.92\textwidth}
\footnotesize
Here $\beta$ denotes the coefficient of $\det X$, and $C$ and $L$ denote the coefficient matrices of the quadratic part and of the linear part $\ell$, respectively. The quantities $g$, $k$ and $N_f$ are as in Propositions~\ref{prop:rank1}--\ref{prop:rank3}, and $N$ denotes the number of roots of $t^2+\alpha_0t-\det L$ in $\Fq$.
\end{minipage}
\end{table}

\begin{proof}
First, if $f\in V_2\oplus V_1\oplus V_0$, then $f=q_2^C+\ell+\alpha_0$ with $C\in M_3(\Fq)$ of rank $r$. Fix a rank-$r$ representative and let $(\ell,\alpha_0)$ range over $V_1\times\Fq$. By Lemma~\ref{lem:V2-transitive}, this gives the weight distribution of Propositions~\ref{prop:linear}--\ref{prop:rank3}. Summing over the $\rho_r$ matrices of rank $r$ gives the stated numbers. The total is $(\rho_0+\rho_1+\rho_2+\rho_3)q^{10}=q^{19}$ codewords.

Now suppose that $f\in\mathcal F(3,6)\setminus(V_2\oplus V_1\oplus V_0)$, so that $\det X$ appears in $f$ with a nonzero coefficient $\beta$. Since multiplication by a nonzero scalar does not change the Hamming weight, we may assume that the coefficient of $\det X$ is $1$. For counting codewords, however, we must multiply the resulting numbers by $q-1$, corresponding to the possible choices of $\beta\in\Fq^\times$.

By Lemma~\ref{lem:translation-elim}, for every $W$, the number of polynomials $\det X+q_2^C+\ell+\alpha_0$ of weight $W$ is $q^9$ times the number of pairs $(\tilde\ell,\tilde\alpha_0)\in V_1\times\Fq$ with $\wt(\det X+\tilde\ell+\tilde\alpha_0)=W$. By Lemma~\ref{lem:rank-only}(1), if $\rank(\tilde\ell)<3$, this weight depends only on $\rank(\tilde\ell)$ and on whether $\tilde\alpha_0$ is zero or nonzero, and the corresponding weights are those computed in Propositions~\ref{prop:cubic-rank0}--\ref{prop:cubic-rank2}. For each rank $r<3$, there are $q^9\rho_r$ codewords corresponding to $\tilde\alpha_0=0$ and $q^9\rho_r(q-1)$ codewords corresponding to $\tilde\alpha_0\ne0$.

If $\rank(\tilde\ell)=3$, then by Lemma~\ref{lem:rank-only}(2) the weight is determined by the number $N$ of roots of $t^2+\tilde\alpha_0t-\det\tilde L$ in $\Fq$, where $\tilde L$ is the coefficient matrix of $\tilde\ell$. By Lemma~\ref{lem:prescribed-det}, for each $c\in\Fq^\times$ there are exactly $\rho_3/(q-1)=|\SL_3(\Fq)|$ matrices of rank $3$ with determinant $c$. Hence, by Lemma~\ref{lem:M-counts}, the number of pairs $(\tilde\ell,\tilde\alpha_0)$ with $\rank(\tilde\ell)=3$ and $N=i$ is $\rho_3M_i/(q-1)$, which gives $q^9\rho_3M_i/(q-1)$ codewords. Finally, multiplying these numbers by $q-1$ to account for the possible nonzero values of $\beta$ gives the stated rows with $\beta\ne0$. Their total number is $(q-1)q^{19}$.

Together with the $q^{19}$ codewords for which $\beta=0$, this gives $q^{19}+(q-1)q^{19}=q^{20}=|\mathcal F(3,6)|$ codewords, as required.
\end{proof}

For $q\ge3$, the $16$ weights listed in Table~\ref{tab:full-wd}, namely the zero weight and $15$ nonzero weights, are pairwise distinct. For $q=2$, the last row of Table~\ref{tab:full-wd} is empty, and $C^{\AA}(3,6)$ has $14$ distinct nonzero weights.

\begin{remark}\label{rem:general}
The weight spectrum of $C^{\AA}(\ell,m)$ for general $\ell$ and $m$ remains open, and Theorem~\ref{thm:full-wd} settles the first case with $\ell\ge3$. The general problem appears to be difficult for the following reasons. In the case of Grassmann codes, the weight of a codeword is constant on the orbits of $\GL_m(\Fq)$ acting on alternating forms, and the determination of the weight spectrum essentially amounts to a classification of such forms. Such a classification is known when $\ell=2$, where the rank is a complete invariant, and for $\ell=3$ in small dimensions. For affine Grassmann codes, the relevant symmetry group is smaller, and hence the corresponding orbits are finer. Consequently, the weight spectrum of an affine Grassmann code cannot be deduced directly from that of the corresponding Grassmann code, even when the latter is known. The present case illustrates this: Nogin~\cite{Nogin1997} showed that the Grassmann code $C(3,6)$ has $5$ distinct nonzero weights, whereas, by Theorem~\ref{thm:full-wd}, the affine Grassmann code $C^{\AA}(3,6)$ has $15$ distinct nonzero weights for $q\ge3$ and $14$ for $q=2$. Moreover, Theorem~\ref{thm:full-wd} shows that, for $\ell=3$, the weight of a codeword is not a function of rank data alone. Indeed, if the determinant occurs with a nonzero coefficient and the linear part has full rank, then the weight depends on the determinant of the coefficient matrix of the linear part and on the constant term, through the number of roots in $\Fq$ of the polynomial $t^2+\alpha_0t-\det L$. Thus, one should not expect a description of the weight spectrum of $C^{\AA}(\ell,m)$ in terms of rank data alone.

On the other hand, our approach does not rely on a classification of alternating forms, and two of its ingredients are not special to $3\times3$ matrices. Namely, the translation argument in the proof of Lemma~\ref{lem:translation-elim} uses only the expansions of $\det(X+U)$ and $\operatorname{adj}(X+U)$, and specializing one row of the generic matrix carries a codeword of $C^{\AA}(\ell,m)$ to a codeword of $C^{\AA}(\ell-1,m-1)$. This suggests an inductive approach on $\ell$, for which \cite{PS2019} and the present article provide the first two steps. The codes $C^{\AA}(3,m)$ with $m>6$ appear to be the natural next case to consider.
\end{remark}

\subsection*{Acknowledgment}
Prasant Singh would like to thank the UGC, Government of India, for the INCP2 grant. This work was initiated during his visit to UiT The Arctic University of Norway, and he gratefully acknowledges UiT for its hospitality. He would also like to thank Sudhir Ghorpade for his valuable inputs and stimulating discussions on this problem.

\bibliographystyle{plain}
\bibliography{Version_2.bib}

\end{document}